\documentclass[10pt]{article}
\usepackage[T1]{fontenc}
\usepackage[a4paper]{geometry}
\usepackage{amssymb}
\usepackage{graphicx,psfig}

\newlength{\algowidth}
\def\uneligne{\centerline{\rule{0.9\textwidth}{0.5pt}}}
\def\algo#1#2{\paragraph{#1}\-\hrulefill\par
\vspace*{-0.5\baselineskip}
            \parbox{0.9\textwidth}{\sf\flushleft
                        #2}\par\vskip 2pt\uneligne\par\vskip 2pt}

\newcommand{\bi}{\begin{itemize}}
\newcommand{\ei}{\end{itemize}}
\newcommand{\bd}{\begin{description}}
\newcommand{\ed}{\end{description}}
\newcommand{\be}{\begin{enumerate}}
\newcommand{\ee}{\end{enumerate}}
\newcommand{\bc}{\begin{center}}
\newcommand{\ec}{\end{center}}

\def\tabulation{xxx\=xxx\=xxx\=xxx\=xxx\=xxx\=xxx\=xxx\=xxx\= \kill}

\def\bibfmta#1#2#3#4{{\textsc{#1}.}\ {#2}, \textit{#3}, #4.}
\def\bibfmtb#1#2#3#4{{\textsc{#1}.}\ \textit{#2}, {#3}, #4.}

\newtheorem{thm}{Theorem}[section]
\newtheorem{lemm}{Lemma}[section]
\newtheorem{defi}{Definition}[section]
\newenvironment{proof}{\begin{trivlist}
                       \item[]\hspace{0cm}\textbf{Proof: }
                       \hspace{0cm} }{\hfill $\square$
                       \end{trivlist}}
                       
\newenvironment{rmk}{\begin{trivlist}
                       \item[]\hspace{0cm}\textbf{Remark}}{
                       \end{trivlist}}
        
\begin{document}

\title{A Uniform Self-Stabilizing \\
Minimum Diameter Spanning Tree Algorithm \\
{\small (Extended Abstract)}}

\author{Franck Butelle\thanks{Universit\'{e} Paris 10 
- 92000 Nanterre. France.},\  
Christian Lavault\thanks{LIPN, CNRS URA 1507, 
Universit\'{e} Paris-Nord\ 93430 Villetaneuse. France.},\ 
Marc Bui\thanks{Universit\'{e} Paris 10 - 92000 Nanterre. France.}
}
\date{\empty}
\maketitle

\begin{abstract}
We present a uniform self-stabilizing algorithm, which solves the
problem of distributively finding a minimum diameter spanning tree of
an arbitrary positively real-weighted graph. Our algorithm consists
in two stages of stabilizing protocols. The first stage is a uniform
randomized stabilizing {\em unique naming} protocol, and the second
stage is a stabilizing {\em MDST} protocol, designed as a {\em fair
composition} of Merlin--Segall's stabilizing protocol and a
distributed deterministic stabilizing protocol solving the (MDST)
problem. The resulting randomized distributed algorithm presented
herein is a composition of the two stages; it stabilizes in
$O(n\Delta+{\cal D}^2+ n \log\log n)$ expected time, and uses $O(n^2
\log n + n \log W)$ memory bits (where $n$ is the order of
the graph, $\Delta$ is the maximum degree of the network, $\cal D$
is the diameter in terms of hops, and $W$ is
the largest edge weight). To our knowledge, our protocol is the very
first distributed algorithm for the (MDST) problem. Moreover, it is
fault-tolerant and works for any anonymous arbitrary network.
\end{abstract}

\section{Introduction}\label{intro}
Many computer communication networks require nodes to broadcast
information to other nodes for network control purposes, which is done
efficiently by sending messages over a spanning tree of the network.
Now optimizing the worst-case message propagation delays over a
spanning tree is naturally achieved by reducing the diameter to a
minimum (see Sect.~\ref{mdstprob}); especially in high-speed
networks (where the message delay is essentially equal to the
propagation delay). However, when communication links fail or come up,
and when processors crash or recover, the spanning tree may have to
be rebuilt. When the network's topology changes, one option is to
perform anew the entire computation of a spanning tree with a minimum
diameter from scratch. We thus examine the question of designing an
efficient fault-tolerant algorithm, which constructs and dynamically
maintains a minimum diameter spanning tree of any anonymous
network. The type of fault-tolerance we require is so-called {\em
``self-stabilization''}, which means, informally, that an algorithm
must be able to ``recover'' from any arbitrary transient fault. In
this setting, we exhibit a self-stabilizing minimum diameter spanning
tree.  Our algorithm is asynchronous, it works for arbitrary anonymous
network topologies (unique processes ID's are not required), it is
uniform (i.e., every process executes the same code; processes are
identical), symmetry is broken by randomization, and it stabilizes in
efficient time complexity.

\subsection{Self-Stabilizing Protocols}\label{ssproto}
We consider distributed networks where processes and links from time
to time can crash and recover (i.e., dynamic networks), where
additionally, when processes recover, their memory may be recovered
within an arbitrary inconsistent state (to model arbitrary memory
corruption). Despite these faults, we wish the network to be able to
maintain and/or to be able to rebuilt certain information about itself
(e.g., in this particular case, maintaining a minimum diameter
spanning tree). When the intermediate period between one recovery and
the next failure is long enough, the system stabilizes.

The theoretical formulation of this model was put forth in the seminal
paper of Dijkstra \cite{Dijk74}, who, roughly, defined the network to
be ``self-stabilizing'' if starting from an {\em arbitrary} initial
state (i.e., after any sequence of faults), the network after some
bounded period of time (denoted as {\em stabilization time}) exhibits
a behaviour as if it was started from a good initial state (i.e,
stabilizes to a ``good'' behaviour, or ``legitimate state''). Notice
that such a formulation does not allow any faults during computation,
but allows an arbitrary initial state. Thus, if new faults occur
during computation, it is modelled in a self-stabilizing formulation
as if it were a {\em new initial state} from which the network again
must recover.
In summary, self-stabilization is a very strong fault-tolerance
property which covers many types of faults and provides a uniform approach
to the design of a variety of fault-tolerant algorithms.

\subsection{The Minimum Diameter Spanning Tree (MDST) Problem}\label{mdstprob}
The use of a control structure spanning the entire network is a
fundamental issue in distributed systems and interconnection networks.
Since {\em all} distributed total algorithms have a time complexity
$\Omega(D)$, where $D$ is the network diameter, a spanning tree of
minimum diameter makes it possible to design a wide variety of time
efficient distributed algorithms.

Let $G=(V(G),E(G))$ be a connected, undirected, positively real-weighted
graph. The {\bf (MDST) problem} is to find a spanning tree of~$G$ of
minimum diameter.

In the remainder of the paper, we denote the problem (MDST), {\em
MDST} denotes the protocol and MDST abbreviates the ``Minimum Diameter
Spanning Tree''.

\subsection{Related Works and Results}
The few literature related to the (MDST) problem mostly deals
either with graph problems in the Euclidian plane (Geometric Minimum
Diameter Spanning Tree), or with the Steiner spanning tree
construction (see~\cite{HLCW91,IhRW91}). The (MDST) problem is clearly
a generalization of the (GMDST) problem. Note that when edge weights are
real numbers (possibly negative), The (MDST) problem is NP-complete.

Surprisingly, although the importance of having a MDST is well-known,
only few papers have addressed the question of how to design
algorithms which construct such spanning trees.  While the problem of
finding and dynamically maintaining a minimum spanning tree has been
extensively studied in the literature (e.g., \cite{Awer87,GaHS83} and
\cite{AwCK90,EITT+92}), there exist no algorithms that
construct and maintain dynamically information about the diameter,
despite the great importance of this issue in the applications. (Very
recently, the distributed (MDST) problem was addressed in
\cite{BuBu93b,Lava95}). In this paper, we present an algorithm which
is robust to transient failures, and dynamically maintains a minimum
diameter spanning tree of any anonymous network: a much more efficient
(computationally cheaper) solution indeed than recomputing from
scratch over and over again.

As opposed to the (quasi-) absence of investigations dealing with the
(MDST) problem, and although self-stabilization is quite a new strand
of research in distributed computing, a large number of
self-stabilizing algorithms and theoretical related results were
proposed during the past few years (e.g.,
\cite{AfBr89,AnEH92,AKMP+93,APVD94,DoIM91a,%
DoIM91b,DoIM93,DoIM95b,KaPe93,ShRR95,Varg94}).  Due to their features,
self-stabilizing protocols were first used in the design of many
existing systems (e.g., DECNET protocols \cite {Perl83}).

Our distributed self-stabilizing algorithm is composed of a first
uniform stabilizing randomized stage protocol {\em UN} of {\em
``unique naming''} for arbitrary anonymous networks and of a second
stabilizing stage protocol {\em MDST}, which constructs a MDST. The
second stage performs a MDST protocol for {\em named} networks which
results after the first stage stabilizes.  This second stage is itself
constructed as the {\em fair composition}
\cite{DoIM93,DoIM95b,ShRR95} of Merlin--Segall's stabilizing
distributed routing protocol and a new deterministic
protocol for the (MDST) problem. The resulting algorithm $\cal A$ is
thus a composition of the two stages (see Sect.~\ref{correct})
to obtain a randomized, uniform, self-stabilizing MDST algorithm $\cal A$
for general anonymous graph systems.

The complexity of protocols is analyzed by the following
complexity measures. The {\bf Time Complexity} of a self-stabilizing
algorithm is mainly defined as the time required for stabilization (or {\em
``round complexity''}). More formally, the {\em stabilization time} of a
self-stabilizing deterministic (resp. randomized) algorithm is the
maximal (resp. maximal expected) number of rounds that takes the
system to reach a legitimate configuration, where the maximum is taken
over all possible executions (see the model $\cal M$ in Sect.~\ref{model}). 
The {\bf Space Complexity} of a self-stabilizing algorithm can be
expressed as the number of bits required to store the state of each
process; i.e., in the message passing model, the maximal size of local
memory used by a process. The {\bf Communication Complexity} is
measured in terms of the number of bits of the registers; i.e., in the
message passing model, the maximal number of bits exchanged by the
processes until an execution of the algorithm stabilizes. The time,
space and communication complexities of a composed algorithm are the
sum of the complexities of the combined protocols.

\subsubsection{Main contributions of the present paper}
\bi
\item A first stage consisting of a uniform stabilizing randomized
{\em UN}  protocol for any arbitrary network~$G$, which is an adapted
variant of the UN protocol designed in \cite{AnEH92}. In model $\cal M$,
our randomized {\em UN} protocol stabilizes in $O(n \log\log n)$
expected time, with a space complexity $O(n^2 \log n)$.
\item An original second stage stabilizing protocol {\em MDST}, which is
designed as the fair composition of Merlin--Segall's stabilizing
routing protocol and a new deterministic protocol for the (MDST)
problem. The second stage thus constructs a MDST of the named network
$G$. In the model $\cal M$, the protocol {\em MDST} stabilizes in
$O(n\Delta + {\cal D}^2)$ time, and its space complexity
is $O(n \log n + n \log W)$ bits (where $\Delta$ is the maximum degree
of~$G$, $\cal D$ is the diameter in terms of hops and $W$ is the
largest edge weight).
\item In model $\cal M$, the resulting randomized composed algorithm
$\cal A$ stabilizes in $O(n\Delta + {\cal D}^2+n\log\log n)$ expected
time and uses
$O(n^2 \log n + n \log W)$ memory bits. To our knowledge, it appears
to be the very first algorithm to {\em distributively} solve the (MDST)
problem. Moreover, our randomized distributed algorithm $\cal A$ is
fault-tolerant and works for any anonymous arbitrary network.
\ei

%\medskip
The remainder of the paper is organized as follows: in
Sect.~\ref{model}, we define the formal model $\cal M$ and
requirements for uniform, self-stabilizing protocols, and in
Sect.~\ref{algo} we present the stages of the composed uniform
self-stabilizing MDST algorithm $\cal A$. Section~\ref{correct} and
Sect.~\ref{anal} are devoted to the correctness proof, and to the
complexity analysis of stabilizing protocols ({\em UN}, {\em MDST},
and algorithm $\cal A$), respectively. The paper ends with concluding
remarks in Sect.~\ref{concl}.

\section{Model $\cal M$ (Message Passing)}\label{model}
Formal definitions regarding Input/Output Automata are omitted from this
abstract \cite{APVD94,Varg94}.

{\bf IO Automata, Stabilization, Time Complexity --}
An Input/Output Automaton (IOA) is a state machine with state transitions
which are given labels called {\em actions}. There are three kinds of actions.
The environment affects the automaton through {\em input actions} which
must be responded to in any state. The automaton affects the environment
through {\em output actions}; these actions are controlled by the automaton
to only occur in certain states. {\em Internal actions} only change the
state of the automaton without affecting the environment.

Formally, an IOA is defined by a {\em state} set $S$, an {\em action} set
$L$, a {\em signature Z} (which classifies $L$ into input, output, and
internal actions), a {\em transition relation} $T \subseteq S \times L \times
S$, and a non-empty set of {\em initial states} $I \subseteq S$. We mostly
deal with {\em uninitialized IOA}, for which $I=S$ ($S$ finite). An action
$a$ is said to be {\em enabled} in state $s$ if there exist $s' \in S$ such
that $(s,a,s') \in T$; input actions are always enabled.
When an IOA ``runs'', it produces an execution. An {\em execution fragment}
is an alternating sequence of states and actions $(s_0,a_1,s_1 \ldots)$,
such that $(s_i,a_i,s_{i+1}) \in T$ for all $i \geq 0$. An execution fragment
is {\em fair} if any internal or output action which is continuously
enabled eventually occurs. An {\em execution} is an execution fragment
which starts with an initial state and is fair. A {\em schedule} is a
subsequence of an execution consisting only of the actions. A {\em
behaviour} is a subsequence of a schedule consisting only of its input
and output actions. Each IOA generates a set of behaviours. Finally, let
$A$ and $B$ denote two IOA, we say that $A$ {\em stabilizes} to $B$ if
every behaviour of $A$ has a suffix which is also a behaviour of $B$.

For time complexity, we assume that every internal or output action which
is continuously enabled occurs in one unit of time. We say that $A$
stabilizes to $B$ in time $t$ if $A$ stabilizes to $B$ and every behaviour
of $A$ has a suffix which occurs within time $t$. The {\em stabilization
time} from $A$ to $B$ is the smallest $t$ such that $A$ stabilizes to $B$
in time $t$.

{\bf Network Model --} The model $\cal M$ is for message passing
protocols. The system is a standard point-to-point asynchronous
distributed network consisting of $n$ communicating processes
connected by $m$ bidirectional links. As usual, the network topology
is described by a connected undirected graph $G = (V,E)$, devoid of
multiple edges and loop-free. $G$ is defined on a set $V$ of vertices
representing the processes and $E$ is a set of edges representing the
bidirectional communication links operating between neighbouring
vertices: in the sequel, $|V| = n$, and $|E| = m$.  We view
communication interconnection networks as undirected graphs.
Henceforth, we use the terms {\em graph} (resp. {\em nodes}/{\em
edges}) and {\em network} (resp. {\em processes}/{\em links})
interchangeably.

Each node and link is modelled by an IOA \cite{APVD94,Varg94}.  A
protocol is {\em uniform} if all processes perform the same protocol
and are indistinguishable; i.e., in our model, we do not assume that
processes have unique identities (ID's). We drop the adjective
``uniform'' from now on. The model $\cal M$ assumes that the messages
are transferred on links in FIFO order, and in a finite but unbounded
delay.  It is also assumed that any non-empty set of processes may
start the algorithm (such starting processes are ``initiators''),
while each non-initiator remains quiescent until reached by some
message.  In model $\cal M$, processes have no global knowledge about
the system (no structural information is assumed), but only know their
neighbours in the network (through the mere knowledge of their ports).
In particular, the model $\cal M$ assumes that nothing is known about
the network size $n$ or the diameter $D(G)$ (no upper bound on $n$ or
on $D(G)$ is either known). Regarding the use of memory, $\cal M$ is
such that the amount of memory used by the protocols remains bounded,
i.e., only a bounded number of messages are stored on each link at any
instant. The justification for this assumption is twofold: first, not
much can be done with unbounded links in a stabilizing setting
\cite{APVD94,DoIM91a,Varg94}, and secondly, real channels are
inherently bounded anyway. In other words, we model bounded links as
unit capacity data links which can store at any given instant at most
one circulating message. A link $uv$ from node $u$ to node $v$ is
modelled as a queue $Q_{uv}$, which can store at most one message from
some message alphabet $\Sigma$ at any instant time. The external
interface to the link $uv$ includes an input action {\sc
Send}$_{uv}(m)$ (``send message $m$ from $u$''), an output action {\sc
Receive}$_{uv}(m)$ (``deliver message $m$ at $v$''), and an output
action {\sc Free}$_{uv}$ (``the link $uv$ is currently free'').  If a
{\sc Send}$_{uv}(m)$ occurs when $Q_{uv}=\emptyset$, the effect is
that $Q_{uv}=\{m\}$; when $Q_{uv}=\emptyset$, {\sc Free}$_{uv}$ is
enabled.  If a {\sc Send}$_{uv}(m)$ occurs when $Q_{uv} \neq
\emptyset$, there is no change of state. Note that by the above timing
assumptions, a message stored in a link will be delivered in one unit
of time.

We refer to \cite{APVD94} for detailed and formal definitions of the
notions of {\em queued node automaton}, {\em network automaton for a
graph G}, and similarly for the notions of {\em internal reset} and
{\em stabilization by local checking and global reset}. (See
Sect.~\ref{correct} for the definition of local checkability and the
statement of the two main theorems used in the correctness proof of the
algorithm).

\section{The Algorithm}\label{algo}
Let $G=(V(G),E(G))$ be a connected, undirected, positively
real-weighted graph, where the weight of an edge $e = uv \in E(G)$ is
given by $\omega_{uv}$. In the remainder of the paper, we use the
graph theoretical terminology and notation. The weight of a path
$[u_{0}, \ldots , u_{k}]$ of~$G$ ($u_{i} \in V(G)$) is defined as
$\sum_{i=0}^{k-1}\omega_{u_{i}u_{i+1}}$. For all nodes $u$ and $v$,
the {\em distance} from $u$ to $v$, denoted $d_G (u,v)$, is the lowest
weight of any path length from $u$ to $v$ in $G$ ($\infty$ if no such
path exists).  The distance $d_G (u,v)$ represents the {\em shortest
path} from $u$ to $v$, and the largest (maximal) distance from node
$v$ to all other nodes in $V(G)$, denoted $s_G (v)$, is the {\em
separation} of node $v$: viz. $s_G (v) =\max_{u\in V(G)} d_G
(u,v)$~\cite{Chri75}.  $D(G)$ denotes the diameter of $G$, defined as
$D(G) = \max_{v\in V(G)} s_G(v)$, and ${\cal D}(G)$ the diameter in
terms of hops. $R(G)$ denotes the radius of $G$, defined as
$R(G)=\min_{v\in V} s_G (v)$. $\Psi_G(u)$ represents a shortest-paths
tree (SPT) rooted at node $u$: $(\forall v\in V(G))\;\;
d_{\Psi_G(u)}(u,v) = d_G(u,v)$.  The set of all SPT's of $G$ is then
denoted $\Psi(G)$. The name of the graph will be omitted when it is clear from the context.

\subsection{A High-Level Description}

\subsubsection{Unique Naming Protocol} \label{un}
The {\em unique naming} protocol solves the (UN) problem, where each
process $u$ must select one ID distinct from all other processes'.
The protocol executes propagation of information (propagation of the
ID of process $u$) and feedback ($u$ collects the ID's of all other
processes): i.e., a ``PIF'' protocol.  Our randomized stabilizing
protocol {\em UN} is a variant of the memory adaptive UN PIF protocol
presented in \cite{AnEH92} and slightly differs in the following
respects.  First, our results hold for the message passing model $\cal
M$, even though they can easily be transposed in the link register
model (and {\em vice versa}: the results in \cite{AnEH92} can easily
be extended to the message passing model).  Next, we do not use the
ranking phase designed in the original protocol, but a simple ID's
conflict checking phase. Besides, our maximum estimate for the size of
the network is arbitrarily chosen to be $\leq\lg n$ (see the proof of
Theorem~\ref{thm:unc} in~\cite{BuBL95}), instead of $n^{1/2} -
n^{1/3}$ in~\cite{AnEH92}.  Note that the model $\cal M$ assumes that
nothing is known about $n$ or $D(G)$ (not even an upper bound),
therefore, the UN Monte-Carlo protocol in
\cite{AnEH92} {\em cannot} be turned into a randomized Las Vegas
protocol (e.g., a protocol solving the (UN) problem with probability~1).

Due to the lack of space, we do not give a detailed description of our
protocol {\em UN} herein. A full description of the three phases executed
in the protocol
can be found in~\cite{AnEH92} (for the original version) and in~\cite{BuBL95}
(for our own variant). However, for better understanding of self-stabilization
(showed in Sect.~\ref{correct}), let us just point out the behaviour of
protocol {\em UN}
in phase~3. Each process in phase~3 repeatedly broadcasts a message
with its ID. At the end of each broadcast, if $u$ detects a conflict, it
initiates a Reset. In addition, $u$ collects the ID's of all other processes
(provided by feedback) and checks that all processes have unique ID's.
The variable $IDList$ contains the list of ID's of the visited processes.
At the beginning of each broadcast, it is set to the initiator's ID;
each visited process attaches its own ID to the list before forwarding it
to its neighbours.
After stabilization, every process remains forever in phase~3.

\subsubsection{Construction of a MDST} \label{mdst}
The definition of separation must be generalized to {\em ``dummy nodes''}
(so-called in contrast to actual vertices of $V$). Such a fictitious
node may possibly be inserted on any edge $e\in E$. Thus, let $e=uv$
be an edge of weight $\omega_{uv}$, a dummy node $\gamma$ inserted on
$e$ is defined by specifying the weight $\alpha$ of the segment
$u\gamma$.  According to the definition, the separation $s(\gamma)$ of
a {\em general node} $\gamma$, whether it is an actual vertex in $V$ or a
dummy node, is clearly given by: $s(\gamma)=\max_{z\in V}d(\gamma,z)$.
A node $\gamma^*$ such that $s(\gamma^*)=\min_{\gamma}s(\gamma)$ is
called an {\it absolute center} of the graph. Recall that $\gamma^*$
always exists in a connected graph, and that is not unique in general.

\begin{figure}[hbt]
\centering
\includegraphics[height=5cm]{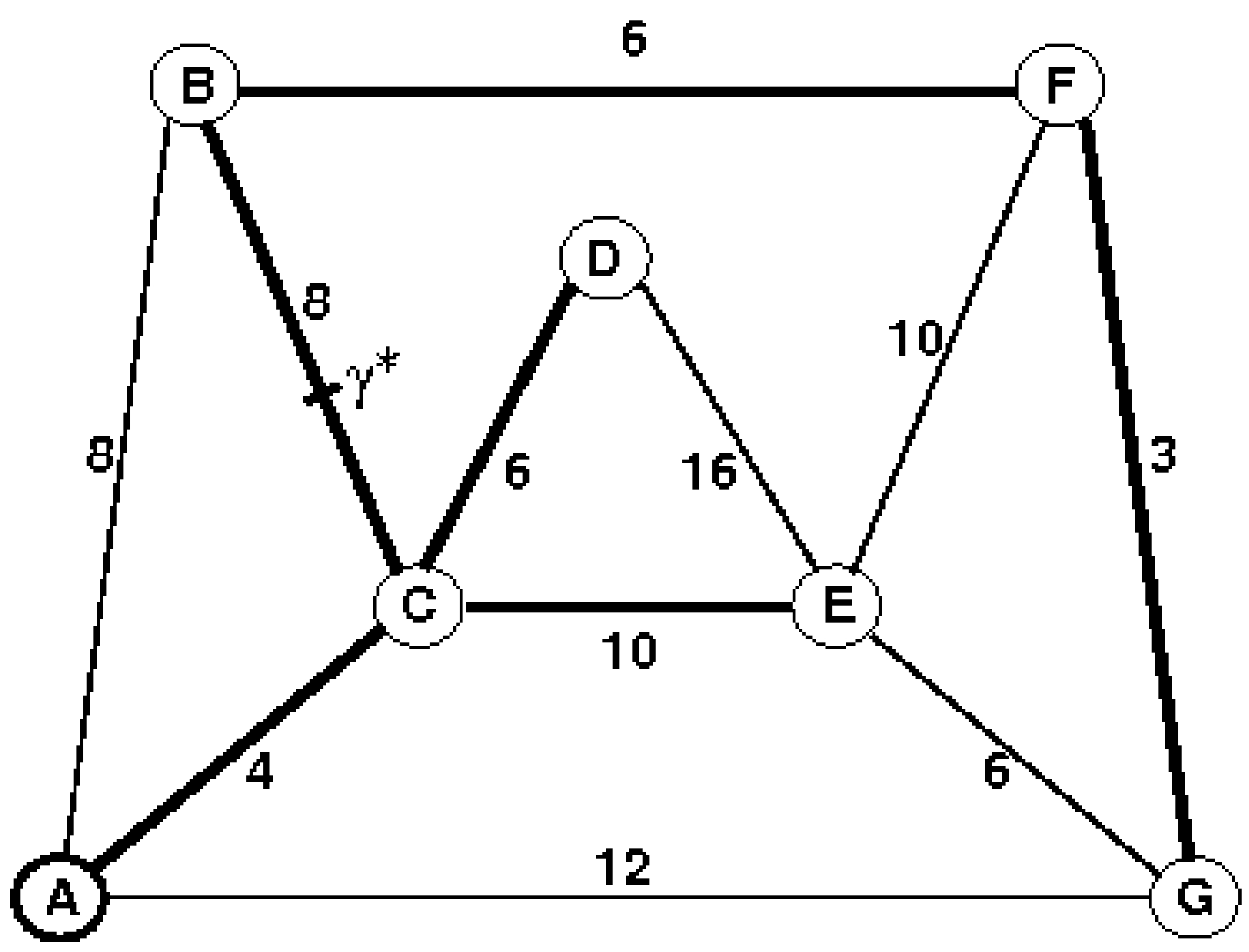}
\caption{Example of a MDST $T^*$ ($D(G)=22$ and $D(T^*)=27$)}
\label{fig:mdst}
\end{figure}

Similarly, the definition of $\Psi(u)$ is also generalized so as to
take these dummy nodes into account. Finding a MDST actually amounts
to search for an absolute center $\gamma^*$ of $G$, and the SPT rooted
at $\gamma^*$ is then a MDST of $G$. Such is the purpose of the
following Lemma:

\begin{lemm}\label{lem:abscenter} {\rm \cite{CaGM80}}
The (MDST) problem for a given graph $G$ is (polynomially) reducible to
the problem of finding an absolute center of $G$.
\end{lemm}

\subsubsection{Computation of an absolute center of a graph} \label{center}\hfill

\noindent
According to the results in~\cite{Chri75}, we use the following Lemma
to find an absolute center of~$G$.

\begin{lemm}\label{lem:hackimi}
Let $G=(V,E)$ be a weighted graph. An absolute center $\gamma^*$
of $G$ is constructed as follows:

\vspace{-0.9\baselineskip}
\be
\item[(i)] On each edge $e \in E$, find a general node $\gamma_e$
of minimum separation.
\item[(ii)] Among all the above $\gamma_e$'s, $\gamma^*$ is a node
achieving the smallest separation.
\ee
\end{lemm}

\begin{proof}(the proof is constructive)

{\em (i)} This first step is performed as follows: for each edge $e=uv$, let
$\alpha = d(u,\gamma)$. Since the distance $d(\gamma,z)$ is the length
of either a path $[\gamma,u,\ldots ,z]$, or a path $[\gamma,v,\ldots ,z]$,
\begin{equation}
s(\gamma)=\max_{z\in V} d(\gamma,z)=\max_{z\in V}\,\min\{\alpha+d(u,z),
\omega_{uv}-\alpha+d(v,z)\}. \label{eq:sep}
\end{equation}

If we plot $f_{z}^+(\alpha)=\alpha+d(u,z)$ and $f_{z}^-(\alpha)=-\alpha +
\omega_{uv}+d(v,z)$ in Cartesian coordinates for fixed $z=z_0$, the
real-valued functions $f_{z_0}^+(\alpha)$ and $f_{z_0}^-(\alpha)$
(separately depending on $\alpha$ in the range $[0,\omega_e]$) are
represented by two line segments $(S_1)_{z_0}$ and $(S_{-1})_{z_0}$,
with slope $+1$ and $-1$, respectively. For a given $z=z_0$, the
smallest of the two terms $f_{z_0}^+(\alpha)$ and $f_{z_0}^-(\alpha)$
(in~(\ref{eq:sep})) is thus found by taking the {\em convex cone} of
$(S_1)_{z_0}$ and $(S_{-1})_{z_0}$. By repeating the above process for
each node $z \in V$, all convex cones of segments $(S_1)_{z \in V}$
and $(S_{-1})_{z \in V}$ are clearly obtained (see
Fig.~\ref{fig:bound}).

Now we can draw the {\em upper boundary} $B_e (\alpha)$ ($\alpha\in
[0,\omega_e]$) of all the above convex cones of segments $(S_1)_{z\in
V}$ and $(S_{-1})_{z \in V}$. $B_e (\alpha)$ is thus a curve made up
of piecewise linear segments, which passes through several local
minima (see Fig.~\ref{fig:bound}). The point $\gamma$ achieving the
smallest minimum value (i.e., the global minimum) of $B_e (\alpha)$
represents the absolute center $\gamma^*_e$ of the edge~$e$.

\medskip
{\em (ii)} By definition of the $\gamma^*_e$'s, $\min_\gamma s(\gamma) =
\min_{\gamma^*_e}s(\gamma^*_e)$, and $\gamma^*$ achieves the
smallest separation. Therefore, an absolute center of the graph is found
at any point where the minimum of all $s(\gamma^*_e)$'s is attained.
\end{proof}

\begin{figure}[htb]
\centering
\includegraphics[height=5cm]{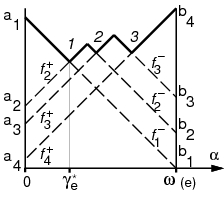}
\caption{Example of an upper boundary $B_e (\alpha)$}
\label{fig:bound}
\end{figure}

By Lemma~\ref{lem:hackimi}, we may consider this method from an
algorithmic viewpoint. For each $e=uv$, let ${\cal C}_e$ be the set of
pairs $\{(d_1,d_2) \; / \; (\forall z \in V) \; d_1=d(u,z),
d_2=d(v,z)\}$ Now, a pair $(d_1',d_2')$ is said to {\em dominate} a
pair $(d_1,d_2)$ iff $d_1\leq d_1'$, and $d_2\leq d_2'$ (viz. the
convex cone of $(d_1',d_2')$ is over the convex cone of $(d_1,d_2)$).
Any such pair $(d_1,d_2)$ will be ignored when it is dominated by
another pair $(d_1',d_2')$.

Notice that the local minima of the upper boundary $B_e (\alpha)$
(numbered from~1 to 3 in Fig.~\ref{fig:bound}) are located at the
intersection of segments $f_{i}^{-}(\alpha)$ and
$f_{i+1}^{+}(\alpha)$, when all dominated pairs are removed. If we
sort the set ${\cal C}_e$ in descending order with respect to the
first term of each remaining pair $(d_1,d_2)$, we thus obtain the list
$L_e=((a_1,b_1), \ldots ,(a_{|L_e|},b_{|L_e|})$ consisting in all such
remaining ordered pairs. Hence, the smallest minimum of $B_e (\alpha)$
for a given edge $e$ clearly provides an absolute center $\gamma^*_e$.
(See Procedure {\tt Gamma\_star($e$)} in Sect.~\ref{formal}). By
Lemma~\ref{lem:hackimi}, once all the $\gamma^*_e$'s are computed,
an absolute center $\gamma^*$ of the graph is obtained. By
Lemma~\ref{lem:abscenter}, finding a MDST of the graph reduces to
the problem of computing~$\gamma^*$.

\subsubsection{All-Pairs Shortest-Paths Protocol (APSP)}\label{apsp}\hfill

%\medskip
\noindent
In the previous paragraph, we consider distances
$d(u,z)$ and $d(v,z)$, for all $z\in V$ and each edge $e=uv$. Such
distances must be computed by a failsafe distributed routing protocol,
e.g., Merlin--Segall's APSP protocol designed in~\cite{MeSe79}.

The justification for this choice is threefold. First, shortest paths
to each destination $v$ are computed by executing the protocol
independently for each $v$. Thus, an essential property of
Merlin--Segall's algorithm is that the routing tables are cycle-free
at any time (Property (a) in \cite{MeSe79}).  Next, the protocol is
also adapted to any change in the topology and the weight of edges
(Property (b)).  Finally, the protocol converges in dynamic networks
and is indeed self-stabilizing (Property (c)). (See
Lemma~\ref{lem:l}).

\subsection{A Formal Description}\label{formal}
Assume the list $L_e$ defined above (in Paragraph~\ref{center}) to be
already constructed (for example with a heap,
whenever the routing tables are computed), the following procedure computes
the value of $\gamma^*_e$ for any fixed edge~$e$.

\algo{Procedure {\tt Gamma\_star($e$)}}{
\begin{tabbing}\tabulation\>
{\bf var} $min, \alpha$ : real \hspace{0.5cm} {\bf Init} $min \leftarrow +
\infty$ ;
$\alpha \leftarrow 0$ ;\+\\
{\bf For} i=1 to $|L_e|$ {\bf do}\+\\
{\sl compute the intersection $(x,y)$ of segments $f_{i}^{-}$ and
$f_{i+1}^{+}$~:}\\
$x=\frac{1}{2}(\omega_e -a_i +b_{i+1})$ ;
$y=\frac{1}{2}(\omega_e +b_{i+1} +a_i)$\\
{\bf if} $y< min$ {\bf then} $min \leftarrow y$ ; $\alpha \leftarrow x$ ;\-\\
{\bf Return}($\alpha$,$min$)
\end{tabbing}
}

%\medskip
The distributed protocol {\em MDST} finds a MDST of an input graph
$G=(V,E)$ by computing the diameter of the SPT's for all nodes.
Initially, an edge weight $\omega_{uv}$ is only known by its two
endpoints $u$ and $v$.  In the first stage, the randomized,
stabilizing protocol {\em UN} provides each process $u$ with its
unique ID, denoted $ID_u$ (see Sect.~\ref{un}).

\algo{Protocol $MDST$ (for process $u$)}{
\begin{tabbing}xxx\=\kill\>
{\bf Type} \=elt : {\bf record} $alpha\_best$,
$upbound$ : real ; $ID_1$, $ID_2$: integer {\bf end} ;\+\\
{\bf Var} \>$\Lambda$ : set of elt ; $\varphi$, $\varphi_u^*$ : elt ;
$D$, $R$, $\alpha$, $localmin$ : real ;\+\\
$d_u$ : array of weights ;\hspace{2cm}{\em (* $d_{u}[v]$ estimates $d(u,v)$ *)}
\end{tabbing}
\vspace{-2mm}
\be
\item {\bf For all} $v\in V$\\
Compute $d_u[v]$, $D$ and $R$ ; \hfill{\em (* by Merlin--Segall's protocol *)}
\item $\varphi.upbound \leftarrow R$ ;
\item {\bf While} $\varphi.upbound>D/2$ {\bf do for} any
edge $uv$ s.t. $ID_v > ID_u$
    \be
    \item $(\alpha,localmin) \leftarrow$ {\tt Gamma\_star($uv$)} ;
    \item {\bf If} $localmin<\varphi.upbound$ {\bf then}
    $\varphi \leftarrow (\alpha,localmin,ID_u,ID_v)$ ;
    \ee
\item $\Lambda \leftarrow \{\varphi\}$ ;
\item {\bf Receive} $\langle \varphi \rangle$ from all sons of $u$
in $\Psi(r)$\\
($r$ is s.t. $ID_r=\min_{v\in V}\{ID_v\}$) ;
$\Lambda\leftarrow \Lambda\cup \{\varphi\}$ ;
\item Minimum finding:
    \be
    \item Compute $\varphi_u^*$ s.t.
$\displaystyle\varphi_u^*.upbound=\min_{\varphi\in\Lambda}\varphi.upbbound$
;\\
    {\bf Send} $\langle \varphi_u^* \rangle$ to father in $\Psi(r)$ ;
    \item {\bf If} $ID_u = ID_r$ {\bf then}
    upon reception of $\langle \varphi\rangle$ from all sons of $r$,
    $r$ forwards $\langle \varphi_u^* \rangle$ to all other nodes.
    \ee
\ee
\vspace{-2mm}
}

\begin{rmk}
In order to complete self-stabilization, the deterministic protocol
{\em MDST} must be repeatedly executed .

A {\em sequential} algorithm for the (MDST) problem may also be
derived from the above protocol, since $\Psi(\gamma)$ is then a MDST of $G$,
where $\gamma$ is the general node s.t. $s(\gamma)=upbound$.
\end{rmk}

\noindent
{\bf Improvements:}\, In practice, some improvements in protocol {\em MDST}
can easily be carried out. Indeed, reducing the enumeration of dummy nodes
may be done by discarding several edges of $G$ from the exploration.
To be able to discard an edge, we only need to know bounds on the minimum
diameter $D^*$ of all spanning trees of $G$. Note that the lower bound on
$D^*$ is obviously $D(G)$, and that $D^*$ is also bounded from above by the
minimum diameter taken over all SPT's, viz. $D^* \leq
\min_{T \in \Psi(G)}D(T)$. In the example of Fig.~\ref{fig:mdst}, such
improvements lead to discard from the exploration the edges
EF, AB, AC, BF, CD, DE, EG, FG. (See~\cite{BuBL95}).

\section{Correctness}
\subsection{Self-Stabilization}\label{ss}
Fix a network automaton $\cal N$ for a given graph $G$, the definition of
local checkability is stated as follows~\cite{APVD94}.

\begin{defi}\label{def:local}
Let ${\cal L} = \{LP_{uv}\}$ be a set of local predicates, and let $\psi$
be any predicate of $\cal N$. A network automaton $\cal N$ is locally
checkable for $\psi$ using $\cal L$ if the following conditions hold.

(i) For all states $s \in S({\cal N})$, if $s$ satisfies $LP_{uv}$ for all
$LP_{uv} \in {\cal L}$, then $s \in \psi$.

(ii) There exists $s \in S({\cal N})$ such that $s$ satisfies $LP_{uv}$
for all $LP_{uv} \in {\cal L}$.

(iii) Each $LP_{uv} \in {\cal L}$ is stable: for all transitions
$(s,a,s')$ of $\cal N$, if $s$ satisfies $LP_{uv}$ then so does $s'$.
\end{defi}

The main theorem in \cite{APVD94} is about self-stabilization by local
checking and global reset. Roughly, it shows that any protocol which
is locally checkable for some global property can be transformed into
an equivalent protocol, which stabilizes to a variant of the protocol
in which the global property holds in its initial state. This
transformation increases the time complexity by an overhead given
in~\cite[Theorem 10]{APVD94}.

Also recall the fundamental Theorem~\ref{thm:comp} which states the
fair composition of two stabilizing protocols $P_1$ and $P_2$ \cite{DoIM93}.
\begin{thm}\label{thm:comp}
If the four conditions hold,

(i) protocol $P_1$ stabilizes to $\psi_1$;

(ii) protocol $P_2$ stabilizes to $\psi_2$ if $\psi_1$ holds;

(iii) protocol $P_1$ does not change variables used by $P_2$ once
$\psi_1$ holds; and,

(iv) all executions are fair w.r.t. both $P_1$ and $P_2$,

then the fair composition of $P_1$ and $P_2$ stabilizes to $\psi_2$.
\end{thm}

\subsection{Correctness Proof}\label{correct}
Let $\psi$ be a predicate over the variables of protocol {\em UN},
and $\psi'$ a predicate over the variables of protocol {\em MDST}
(see Sect.~\ref{model}). Now, protocol {\em MDST} is the fair combination
of two subprotocols. The first protocol uses Merlin--Segall's APSP routing
algorithm (see Sect.~\ref{apsp}),
while the second subprotocol deterministically
computes the value $\gamma^*$ (see Sect.~\ref{center}). Hence, the local
predicates $LP_{uv}$ and $LP'_{uv}$ corresponding to the predicates
$\psi$ and $\psi'$, respectively, are defined by
\begin{eqnarray*}
LP_{uv} & \equiv & \left\{(\forall \; ID_i, ID_j \in IDList_u )\;\; i\not= j\;
\Longrightarrow \; ID_i \not= ID_j\right\}\\
& &\mbox{} \land  \;\; \left\{(\forall \;
ID_i, ID_j \in IDList_v) \;\; i\not = j\;
\Longrightarrow\; ID_i\not= ID_j\right\}\\
& & \mbox{} \land \;\; (u\ and\ v\ are\ both\ in\ phase\ 3) \; \;
\mbox{\rm for predicate}\ \; \; \psi\equiv (\forall uv \in E) \;\; LP_{uv},
\end{eqnarray*}
where the variable $IDList$ is defined in Sect.~\ref{un}.
And, similarly,
$$LP'_{uv} \equiv (d_u[v] < +\infty) \; \land \; (d_v[u] < +\infty),\; \;
\mbox{\rm for predicate}\ \; \; \psi' \equiv (\forall uv \in E) \; \; LP'_{uv}.$$

Note that this does not mean that the estimate values $d_u[v]$ are exact,
but that they are not too bad. Of course, if some distances $d_u[v]$ are
wrong, it may cause the construction of a MDST to fail. However, the
routing protocol is self-stabilizing, and after a while the estimate
distances shall be correct and a MDST will be found.

\begin{lemm}\label{lem:l}
Let ${\cal L}=\{LP_{uv}\}$ be the set of local predicates over the
variables of the randomized protocol UN. A network automaton $\cal N$
is locally checkable for $\psi$ using $\cal L$.
\end{lemm}

\begin{proof}
(By Definition \ref{def:local}).
Condition {\em (i)} clearly holds by the definition of $\psi$.
Condition {\em (ii)} holds for a state $s\in S({\cal N})$ such that
processes ID's are all distinct in phase~3.

Now suppose $s\in S({\cal N})$ satisfies $LP_{uv}$. In the case when
no failures occur, $u$ and $v$ obviously remain in phase~3 by construction
of protocol {\em UN}. In the case when nodes recoveries occur (with
arbitrary ID's), $u$ and $v$ are able to detect conflicts and if
necessary they initiate a Reset. After a while, each process (and
especially $u$ and $v$) returns to phase~3 with one unique ID.
Therefore, condition {\em (iii)} holds.
\end{proof}

\begin{lemm}\label{lem:l'}
Let ${\cal L'}=\{LP'_{uv}\}$ be the set of local predicates over the variables
of Merlin--Segall's APSP protocol.
A network automaton $\cal N$ is locally checkable for $\psi'$ using $\cal L'$.
\end{lemm}

\begin{proof}
(By Definition \ref{def:local}).
Condition {\em (i)} clearly holds by the definition of $\psi'$.
Since $G$ is connected, there exists a path $[u,\ldots ,v]$ such that the
distance $d_G(u,v)$ is finite. Hence, condition {\em (ii)} holds
for the corresponding state $s\in S({\cal N})$.
Finally, condition {\em (iii)} clearly holds by convergence
of Merlin--Segall's routing protocol. (See \cite[Property~(c)]{MeSe79}, and
Sect.~\ref{apsp}).
\end{proof}
Recall that $\varphi_u^*.upbound$ denotes the best value of $s(\gamma^*)$
computed so far at node $u$. We show now that both protocols {\em MDST} and $\cal A$ stabilize to the
desired postcondition $\theta$ defined by:
$$\theta \equiv (\forall u \in V) \;\; \varphi_u^*.upbound=s(\gamma^*).$$
The local predicate $LP''_{uv}$ corresponding to $\theta$ is defined by:
$$LP''_{uv}\equiv \varphi_u^*.upbound =s(\gamma^*)
\;\land\;\varphi_v^*.upbound =s(\gamma^*).$$

\begin{lemm}\label{lem:mdst-stab}
Assume processes ID's are all distinct, the protocol MDST stabilizes to
$\theta$.
\end{lemm}

\begin{proof} (Sketch) First, protocol {\em MDST} is locally checkable for
$\theta$ using the set ${\cal L''}=\{LP''_{uv}\}$.
By Definition~\ref{def:local}, conditions {\em (i)} and {\em (ii)} clearly
hold.
Condition {\em (iii)} derives from the fact that Merlin--Segall's
protocol stabilizes to $\psi'$, while the computation of $\gamma^*$ is
deterministic. Consequently, protocol {\em MDST} is locally checkable
and stabilizes to $\theta$ by \cite[Theorem~10]{APVD94}.
\end{proof}

\begin{thm}\label{thm:algostab}
The randomized algorithm $\cal A$ stabilizes to $\theta$ with probability~1.
\end{thm}

\begin{proof} The following conditions hold.

{\em (i)} By Lemma~\ref{lem:l} and \cite[Theorem~10]{APVD94},
protocol {\em UN} stabilizes to $\psi$ with probability~1.

{\em (ii)} By Lemma~\ref{lem:mdst-stab},
protocol {\em MDST} stabilizes to $\theta$ if $\psi$ holds.

{\em (iii)} By construction,
protocol {\em UN} does not change variables used by {\em MDST}
once $\psi$ holds.

{\em (iv)} Since protocol {\em MDST} terminates,
there are only finitely many executions of {\em MDST} between two
executions of {\em UN}. The protocol {\em UN} stabilizes to
$\theta$ with probability~1
and since $\theta$ is true, each ID remains unchanged,
and so does the computation of $\gamma^*$. Therefore,
all executions are fair w.r.t. to both {\em UN} and {\em MDST}.

By Theorem~\ref{thm:comp}, algorithm $\cal A$ which is the
fair composition of {\em UN} and {\em MDST} stabilizes to $\theta$
with probability~1.
\end{proof}

\section{Analysis}\label{anal}

\subsection{Protocol $UN$}
The three phases executed in protocol {\em UN} are described in
\cite{AnEH92,BuBL95}. (See Sect.~\ref{un}).

\begin{lemm}\label{lem:rounds}
Each Reset lasts at most $2D + n$ rounds. If two processes have the
same ID, then within at most $O(n)$ rounds some process in the network
will order a Reset. After a Reset, it takes the system $4n$ rounds
either to perform global memory adaptation, or to complete another
Reset.
\end{lemm}
Note that the maximum number of rounds needed for the completion of phases
1 and 2 is exactly $4n$.
\begin{lemm}\label{lem:ID's}
If $n$ processes choose random ID's from the set $[N] = \{1, \ldots ,N\}$,
where $N \geq n^{2}/{\epsilon}$, all ID's will be unique with probability
$p >1 - \epsilon$, for all $0 < \epsilon < 1$.
\end{lemm}

\begin{proof}
The probability $p$ that all processes randomly choose distinct ID's is
$$p \: = \: \frac{N (N - 1) \cdots (N - n + 1)}{N^{n}}
\: = \: \prod_{i=1}^{n-1} (1 - i/N).$$
Assuming that $n/N \leq 1/2$, or $N \geq 2n$ yields
$$p \: > \: \prod_{i=1}^{n-1} e^{-2i/N} \: >\: e^{-n^{2}/N}.$$
Since $(\forall \: 0 < \epsilon < 1) \; e^{-\epsilon} > 1-\epsilon$, we
have that $p > 1 - \epsilon$ when $N \geq n^{2}/{\epsilon}$. Hence,
it is sufficient to randomly select the $n$ identities from the set $[N]$,
with $N \geq n^{2}/{\epsilon}$, in which case the identities are all
distinct with probability $> 1 - \epsilon$, for fixed $0 < \epsilon < 1$.
%In the worst-case where $N \rightarrow \infty, p$ tends to 1.
\end{proof}

\begin{lemm}\label{lem:reset}
If after a Reset there exist $n' < n$ distinct ID's in the network, then a
Reset is initiated by the end of phase 2 with probability $\geq 1- 2^{n'-n}$.
\end{lemm}

\sloppypar
\begin{thm}\label{thm:unc}
Let $\Delta$ be the maximum degree of the network. Starting from any state,
the probability that the system will stabilize in $O\left(n (1 + \log\log n
- \log\log(\Delta+1)) \right)$ rounds is $\ge 1 - \delta$, for some constant
$0<\delta <1$ which does not depend on the network. The expected number
of rounds until protocol UN stabilizes is $O(n \log\log n)$. The maximal
memory size used by each process in any execution of protocol UN is at most
$O(n^2 \log n)$ bits.
\end{thm}

\subsection{Protocol $MDST$}\label{mdstcomp}

\begin{lemm}\label{lem:mdstc}
The time complexity of protocol MDST is at most
$O(n\Delta+{\cal D}^2)$, and its space complexity is
$O(n\log n + n\log W)$ bits, where $W$ is the largest edge weight.
\end{lemm}

\begin{proof}
It is shown in~\cite{MeSe79} that after $i$ update
rounds, all shortest paths of at most $i$ hops have been correctly
computed, so that after at most $\cal D$ rounds, all shortest paths to
node $u$ are computed.  Shortest paths to each destination are
computed by executing the protocol independently for each destination.
Since a round costs $O({\cal D})$ time, the stabilization time of
Merlin--Segall's protocol is $O({\cal D}^2)$.  Now, the computation of
$\gamma^*$ requires a minimum finding over a tree (viz., $O(n)$) and
local computations on each adjacent edge of $G$ (viz., $O(\Delta)$,
where $\Delta$ is the maximum degree). Hence, the stabilization time
of the protocol {\em MDST} is $O(n\Delta + {\cal D}^2)$.

Finally, $O(n\log n + n\log W)$ space complexity is needed to maintain
global routing tables.
\end{proof}
Note that since ${\cal D} \leq D \leq W {\cal D}$, the ``hop time complexity''
used above is more accurate.

\subsection{Complexity Measures of Algorithm $\cal A$}

The following theorem summarizes our main result, and its proof follows from
the previous Lemma.

\begin{thm}\label{thm:complex}
Starting from any state, the probability that algorithm $\cal A$ will
stabilize is $\geq 1-\delta$, for some constant $0<\delta <1$ which
does not depend on the network. Recall ${\cal D}$ be the diameter
of~$G$ in terms of hops, $\Delta$ the maximum degree, and $W$ the
largest edge weight. The expected time complexity of $\cal A$ is
$O(n\Delta + {\cal D}^2 + n \log\log n)$, and its space complexity is
at most $O(n^2\log n + n\log W)$ bits.
\end{thm}

Since the number of messages required in Merlin--Segall's protocol is
at most $O(n^2m)$, the communication complexity of $\cal A$ is
$O(n^2mK)$ bits (where $K=O(\log n +\log W)$ bits is the largest
message size).

\section{Concluding Remarks}\label{concl}
We proposed a uniform self-stabilizing algorithm for distributively
finding a MDST of a positively weighted graph. Our algorithm is new.
It works for arbitrary anonymous networks topologies, symmetry is
broken by randomization; it stabilizes in $O(n\Delta + {\cal D}^2 + n
\log\log n)$ expected time, and requires at most $O(n^2\log n + n\log
W)$ memory bits. The assumptions of our model $\cal M$ are quite
general, and in some sense, the algorithm might be considered
reasonably efficient in such a setting (even though the communication
complexity appears to be the weak point of such algorithms).
Whatsoever, the stabilization complexities can be improved in terms
of time and space efficiency by restricting the model's assumptions
and using the very recent results proposed in \cite{Dole94} and
\cite{AKMP+93}. First, the randomized uniform self-stabilizing
protocol presented in \cite{Dole94} provides each (anonymous) process
of a uniform system with a distinct identity. This protocol for unique
naming uses a predefined fixed amount of memory and stabilizes within
$\Theta(D)$ expected time (where $D$ is the diameter of the network).
Secondly, following \cite{AKMP+93}, we may restrict our model and assume
that a pre-specified bound $B(D)$ on the diameter $D$ is known.
In $O(D)$ time units, the stabilizing protocol in \cite{AKMP+93} produces
a shortest paths tree rooted at the minimal ID node of the network;
in addition, the complexity of the space requirement and messages size
is $O(\log B(D))$. In this restricted model (i.e., assuming the knowledge
of an upper bound on $D$), the fair composition of the two protocols yields
a randomized uniform self-stabilizing algorithm which finds a MDST with
stabilization time (at most) $O(n)$ and space complexity $O(\log B(D))$. 
In this setting, the fact that the space complexity does {\em not}
depend on $n$ makes the solution more adequate for dynamic networks.

\end{document}